\documentclass[journal]{IEEEtran}

\usepackage{graphicx,url}
\usepackage{dcolumn}
\usepackage{bm}
\usepackage{amsmath,amsfonts,amssymb,citesort}

\newtheorem{definition}{\noindent{\it Definition}}[section]
\newtheorem{theorem}{\noindent{\it Theorem}}[section]

\newtheorem{lemma}[theorem]{\noindent{\it Lemma}}

\newtheorem{remark}[theorem]{\noindent{\it Remark}}

\newenvironment{proof}{\noindent{\it Proof:}}{$\hfill$ $\Box$\\ }

\begin{document}

\title{Convolutional Codes: Techniques of Construction}

\author{Giuliano G. La Guardia
\thanks{Giuliano Gadioli La Guardia is with Department of Mathematics and Statistics,
State University of Ponta Grossa (UEPG), 84030-900, Ponta Grossa,
PR, Brazil. }}

\maketitle

\begin{abstract}
In this paper we show how to construct new convolutional codes
from old ones by applying the well-known techniques: puncturing,
extending, expanding, direct sum, the $({ \bf u}| { \bf u}+{ \bf
v})$ construction and the product code construction. By applying
these methods, several new families of convolutional codes can be
constructed. As an example of code expansion, families of
convolutional codes derived from classical
Bose-Chaudhuri-Hocquenghem (BCH), character codes and Melas codes
are constructed.
\end{abstract}

\section{Introduction}\label{Intro}

Constructions of (classical) convolutional codes and their
corresponding properties have been presented in the literature
\cite{Forney:1970,Piret:1983,Piret:1988,Piret:1988art,Rosenthal:1999,York:1999,Hole:2000,Rosenthal:2001,Gluesing:2006,Luerssen:2006,Schmale:2006,Aly:2007,Luerssen:2008,LaGuardia:2012}.
Forney \cite{Forney:1970} was the first author who introduced
algebraic tools in order to describe convolutional codes.
Addressing the construction of maximum-distance-separable (MDS)
convolutional codes (in the sense that the codes attain the
generalized Singleton bound introduced in \cite[Theorem
2.2]{Rosenthal:1999}), there exist interesting papers in the
literature \cite{Rosenthal:1999,Rosenthal:2001,Schmale:2006}.
Concerning the optimality with respect to other bounds we have
\cite{Piret:1983,Piret:1988art}, and in \cite{Gluesing:2006},
strongly MDS convolutional codes were constructed. In
\cite{York:1999,Hole:2000,Aly:2007,LaGuardia:2012}, the authors
presented constructions of convolutional BCH codes. In
\cite{Luerssen:2006}, doubly-cyclic convolutional codes were
constructed and in \cite{Luerssen:2008}, the authors described
cyclic convolutional codes by means of the matrix ring.

It is not simple to derive families of such codes by means of
algebraic approaches. In other words, most of the convolutional
codes available in the literature are constructed case by case. In
order to derive new families of convolutional codes by means of
algebraic methods, this paper is devoted to construct new
convolutional from old ones. More precisely, we show how to obtain
new codes by extending, puncturing, expanding, applying the direct
sum, applying the $({ \bf u}| { \bf u}+{ \bf v})$ construction and
finally employing the product code construction.

This paper is arranged as follows. In Section~\ref{nota} we fix
the notation. Section~\ref{II} presents a review of basics on
convolutional codes. In Section~\ref{III} we present the
contributions of this work, namely, we show how to construct
convolutional codes by applying the techniques above mentioned. In
Section~\ref{IV} we present two code tables containing parameters
of convolutional codes known in the literature as well as
containing parameters of new convolutional codes obtained from the
code expansion developed in Subsection~\ref{sub3.1}; in
Section~\ref{V}, a summary of this paper is given.

\section{Notation}\label{nota}

Throughout this paper, $p$ denotes a prime number, $q$ denotes a
prime power, ${\mathbb F}_{q}$ is a finite field with $q$
elements. The (Hamming) distance of two vectors ${ \bf v}, { \bf w
}\in {\mathbb F}_{q}^{n}$ is the number of coordinates in which ${
\bf v}$ and ${ \bf w }$ differ. The (Hamming) weight of a vector
${ \bf v}=(v_1, v_2, \ldots, v_{n}) \in {\mathbb F}_{q}^{n}$ is
the number of nonzero coordinates of ${ \bf v}$. The trace map
tr$_{q^{m}/q}: {\mathbb F}_{q^{m}} \longrightarrow {\mathbb
F}_{q}$ is defined as tr$_{q^{m}/q}(a):= \displaystyle
\sum_{i=0}^{m-1} a^{q^{i}}$. Sometimes we abuse the notation by
writing $C = {[n, k, d]}_{q}$ to denote a liner code of length
$n$, dimension $k$ and minimum distance $d$. The (Euclidean) dual
of $C$ is denoted by $C^{\perp}$ and we put
$d^{\perp}=d(C^{\perp})$ meaning the minimum distance of
$C^{\perp}$.

\section{Convolutional Codes}\label{II}

In this section we present a brief review of convolutional codes.
For more details we refer the reader to
\cite{Piret:1988,Johannesson:1999,Huffman:2003}.

Recall that a polynomial encoder matrix $G(D) \in {\mathbb
F}_{q}{[D]}^{k \times n}$ is called \emph{basic} if $G(D)$ has a
polynomial right inverse. A basic generator matrix of a
convolutional code $C$ is called \emph{reduced} (or minimal
\cite{Rosenthal:2001,Huffman:2003,Luerssen:2008}) if the overall
constraint length $\delta =\displaystyle\sum_{i=1}^{k} {\delta}_i$
has the smallest value among all basic generator matrices of $C$;
in this case the overall constraint length $\delta$ is called the
\emph{degree} of the code. The \emph{weight} of an element ${\bf
v}(D)\in {\mathbb F}_{q} {[D]}^{n}$ is defined as wt$({\bf
v}(D))=\displaystyle\sum_{i=1}^{n}$wt$(v_i(D))$, where
wt$(v_i(D))$ is the number of nonzero coefficients of $v_{i}(D)$.

 \begin{definition}\cite{Johannesson:1999}
A rate $k/n$ \emph{convolutional code} $C$ with parameters $(n, k,
\delta ; \mu,$ $d_{f} {)}_{q}$ is a submodule of ${\mathbb F}_q
{[D]}^{n}$ generated by a reduced basic matrix $G(D)=(g_{ij}) \in
{\mathbb F}_q {[D]}^{k \times n}$, that is, $C = \{ {\bf u}(D)G(D)
| {\bf u}(D)\in {\mathbb F}_{q} {[D]}^{k} \}$, where $n$ is the
length, $k$ is the dimension, $\delta =\displaystyle\sum_{i=1}^{k}
{\delta}_i$ is the \emph{degree}, where ${\delta}_i =
{\max}_{1\leq j \leq n} \{ \deg g_{ij} \}$, $\mu = {\max}_{1\leq
i\leq k}\{{\delta}_i\}$ is the \emph{memory} and
$d_{f}=$wt$(C)=\min \{wt({\bf v}(D)) \mid {\bf v}(D) \in C, {\bf
v}(D)\neq 0 \}$ is the \emph{free distance} of the code.
\end{definition}

If ${\mathbb F}_{q}((D))$ is the field of Laurent series we define
the weight of ${\bf u}(D)$ as wt$({\bf u}(D)) = {\sum}_{i \in
\mathbb{Z}}$wt$(u_i)$. A generator matrix $G(D)$ is called
\emph{catastrophic} if there exists a ${\bf u}{(D)}^{k}\in
{\mathbb F}_{q}{((D))}^{k}$ of infinite Hamming weight such that
${\bf u}{(D)}^{k}G(D)$ has finite Hamming weight. The
convolutional codes constructed in this paper have basic generator
matrices; consequently, they are noncatastrophic.

We define the Euclidean inner product of two $n$-tuples ${\bf
u}(D) = {\sum}_i {\bf u}_i D^i$ and ${\bf v}(D) = {\sum}_j {\bf
v}_j D^j$ in ${\mathbb F}_q {[D]}^{n}$ as $\langle {\bf u}(D)\mid
{\bf v}(D)\rangle = {\sum}_i {\bf u}_i \cdot {\bf v}_i$. If $C$ is
a convolutional code then its (Euclidean) dual is given by
$C^{\perp }=\{ {\bf u}(D) \in {\mathbb F}_q {[D]}^{n}\mid \langle
{\bf u}(D)\mid {\bf v}(D)\rangle = 0$ for all ${\bf v}(D)\in C\}$.

\subsection{Convolutional Codes Derived From Block
Codes}\label{IIA}

Let ${[n, k, d]}_{q}$ be a block code whose parity check matrix
$H$ is partitioned into $\mu +1$ disjoint submatrices $H_i$ such
that \begin{eqnarray*} H = \left[
\begin{array}{c}
H_0\\
H_1\\
\vdots\\
H_{m}\\
\end{array}
\right], \ \ \ \ (1)
\end{eqnarray*}
where each $H_i$ has $n$ columns, obtaining the polynomial matrix
\begin{eqnarray*}
G(D) =  {\tilde H}_0 + {\tilde H}_1 D + {\tilde H}_2 D^2 + \cdots
+ {\tilde H}_{\mu} D^{\mu}. \ \ \ \ (2)
\end{eqnarray*}
The matrix $G(D)$ with $\kappa$ rows generates a convolutional
code $V$, where $\kappa$ is the maximal number of rows among the
matrices $H_i$; the matrices ${\tilde H}_i$, where $0\leq i\leq
\mu$, are derived from the respective $H_i$ by adding zero-rows at
the bottom in such a way that the matrix ${\tilde H}_i$ has
$\kappa$ rows in total.

\begin{theorem}\cite[Theorem 3]{Aly:2007}\label{A}
Suppose that $C \subseteq {\mathbb F}_q^n$ is an ${[n, k, d]}_{q}$
code with parity check matrix $H \in {\mathbb F}_q^{(n-k)\times
n}$ partitioned into submatrices $H_0, H_1, \cdots, H_{\mu}$ as
above, such that $\kappa =$ $\operatorname{rk} H_0$ and
$\operatorname{rk} H_i \leq \kappa$ for $1 \leq i\leq \mu$. Then
the matrix $G(D)$ is a reduced basic generator matrix. Moreover,
if $d_f$ and $d_f^{\perp}$ denote the free distances of $V$ and
$V^{\perp}$, respectively, $d_i$ denote the minimum distance of
the code $C_i = \{ {\bf v}\in {\mathbb F}_q^n \mid {\bf v} {\tilde
H}_i^t =0 \}$ and $d^{\perp}$ is the minimum distance of
$C^{\perp}$, then one has $\min \{ d_0 + d_{\mu} , d \} \leq
d_f^{\perp} \leq  d$ and $d_f \geq d^{\perp}$.
\end{theorem}


\section{Construction Methods}\label{III}

Constructions of convolutional codes have appeared in the
literature
\cite{Piret:1988,Rosenthal:1999,Hole:2000,Rosenthal:2001,Gluesing:2006,Luerssen:2006,Schmale:2006,Luerssen:2008}.
In this section we apply the method proposed by Piret
\cite{Piret:1988} and recently generalized by Aly \emph{et al.}
\cite[Theorem 3]{Aly:2007}, which consists in the construction of
convolutional codes derived from block codes, in order to
generalize to convolutional codes the techniques of puncturing,
extending, expanding, direct sum, the $({ \bf u}| { \bf u}+{ \bf
v})$ construction and the product code construction. Throughout
this section we utilize the notation given in Section~\ref{II} (
Subsection~\ref{IIA}).

\begin{remark}
Although most families of convolutional codes constructed in this
section consist of unit memory codes, the procedures adopted in
this paper also hold in more general setting, i.e., the techniques
developed here also hold to construct families of multi-memory
convolutional codes as well.
\end{remark}

\subsection{Code Expansion}\label{sub3.1}

Given a basis $\beta=\{ b_1, b_2, \ldots, b_{m} \}$ of ${\mathbb
F}_{q^{m}}$ over ${\mathbb F}_{q}$, a \emph{dual basis}
\cite{Lidl:1997} of $\beta$ is given by ${\beta}^{\perp}=\{
{b_1}^{\ast}, {b_2}^{\ast},\ldots, {b_{m}}^{\ast} \}$, with
tr$_{q^{m}/q}(b_{i}{b_{j}}^{\ast}) =\delta_{ij}$, for all $i,j\in
\{1,\ldots, m \}$. If $C$ is an ${[n, k, d_1]}_{q^{m}}$ code and
$\beta=\{b_1, b_2, \ldots, b_{m}\}$ is a basis of ${\mathbb
F}_{q^{m}}$ over ${\mathbb F}_{q}$, then the $q$-ary expansion
$\beta(C)$ of $C$ with respect to $\beta$ is an ${[mn, mk, d_{2}
\geq d_1]}_{q}$ code given by $\beta(C):= \{ {(c_{ij})}_{i,j}\in
{{\mathbb F}_{q}}^{mn} \mid\textbf{c}= {(\sum_{j}^{} c_{ij} b_{j}
{)}_{i}} \in C \}$. We need the following result to proceed
further.

\begin{lemma}\label{dualcode}\cite{Grassl:1999,Ashikhmin:2000,LaGuardia:2010}
Let $C={[n, k, d]}_{q^{m}}$ be a linear code over ${\mathbb
F}_{q^{m}}$, where $q$ is a prime power. Let ${C}^{\perp}$ be the
dual of the code $C$. Then the dual code of the $q$-ary expansion
$\beta(C)$ of code $C$ with respect to the basis $\beta$ is the
$q$-ary expansion ${{\beta}^{\perp}}({C}^{\perp})$ of the dual
code ${C}^{\perp}$ with respect to ${\beta}^{\perp}$.
\end{lemma}

Theorem~\ref{MAINI} presents a method to construct convolutional
codes by expanding classical codes:

\begin{theorem}\label{MAINI}
Suppose that there exists an $(n, k, \delta ; 1, d_{f}
{)}_{q^{m}}$ convolutional code $V$ derived from an ${[n, k^{*},
d]}_{q^{m}}$ code $C$, where $d_{f}\geq d^{\perp}$. Then there
exists an $(nm, km, m\delta ; 1, d_{f_{\beta}} {)}_{q}$
convolutional code, where $d_{f_{\beta}} \geq d^{\perp}$.
\end{theorem}
\begin{proof}
Let $V$ be the convolutional code generated by the reduced basic
generator matrix $G(D)={H}_0 + {\tilde H}_1 D$, where $$H = \left[
\begin{array}{c}
H_0\\
H_1\\
\end{array}
\right]$$ is a parity check matrix of $C$ and $\operatorname{rk}
H_0 = k \geq \operatorname{rk} {H}_1 = \operatorname{rk} {\tilde
H}_1 = \delta$. We expand $C^{\perp}$ by the dual basis
${\beta}^{\perp}$ of ${\mathbb F}_{q^{m}}$ over ${\mathbb F}_{q}$,
generating the linear code ${\beta}^{\perp} (C^{\perp})$ with
parameters ${[mn, m(n-k^{*}), d_{\beta}\geq d^{\perp}]}_{q}$. From
Lemma~\ref{dualcode} we know that ${\beta}^{\perp} (C^{\perp}) =
{[\beta(C)]}^{\perp}$. We can split the generator matrix
$H_{\beta}$ of ${[\beta(C)]}^{\perp}$ in the form $$H_{\beta} =
\left[
\begin{array}{c}
H_0^{\beta}\\
H_1^{\beta}\\
\end{array}
\right],$$ where $\operatorname{rk} H_0^{\beta} = mk$ and
$\operatorname{rk} {H}_1^{\beta} = m\delta$, because
$mk+m\delta=m(n-k^{*})$. Consider the polynomial matrix
$${[G(D)]}_{\beta} =  {H}_0^{\beta} + {\tilde H}_1^{\beta} D.$$ From
construction and from Theorem~\ref{A}, ${[G(D)]}_{\beta}$ is a
reduced basic generator matrix of a convolutional code
$V_{\beta}$. From construction, $V_{\beta}$ has parameters $(nm,
km, m\delta ; 1, d_{f_{\beta}} {)}_{q}$. From Theorem~\ref{A}, we
have $d_{f_{\beta}}\geq d^{\perp}$, so the result follows. (Note
that the code $\beta(C)$ constructed in this proof plays the role
of the code $C$ given in Theorem~\ref{A}).
\end{proof}

\subsection{Direct Sum Codes}\label{sub3.2}

Let us recall the direct sum of codes. For more details the reader
can consult \cite{Macwilliams:1977,Huffman:2003} for example.
Assume that $C_1={[n_{1}, k_{1}, d_{1} ]}_{q}$ and $C_{2}={[n_{2},
k_{2}, d_{2}]}_{q}$ are two linear codes. Then the direct sum code
$C_{1}\oplus C_{2}$ is the linear code given by $C_{1}\oplus
C_{2}=\{ ({{\bf c}}_{1}, {{\bf c}}_{2})| {{\bf c}}_{1} \in C_{1},
{{\bf c}}_{2}\in C_{2} \}$ and has parameters $[n_1+n_{2},
k_{1}+k_{2}, \min$ $\{d_{1},d_{2}\} {]}_{q}$.

Let us now prove the main result of this subsection:

\begin{theorem}\label{MAINII}
Assume that there exist an $(n_1, k_1, {\delta}_1 ; 1, d_{f}^{(1)}
{)}_{q}$ convolutional code $V_1$ derived from an ${[n_{1},
k_{1}^{*}, d_{1} ]}_{q}$ code $C_1$, where $d_{f}^{(1)}\geq
d_{1}^{\perp}$, and an $(n_{2}, k_{2}, {\delta}_{2} ; 1,
d_{f}^{(2)} {)}_{q}$ convolutional code $V_{2}$ derived from an
${[n_{2}, k_{2}, d_{2}]}_{q}$ code $C_{2}$, where $d_{f}^{(2)}\geq
d_{2}^{\perp}$. Then there exists an $(n_1 + n_{2}, k_1 + k_{2},
{\delta}_{1}+ {\delta}_{2} ; 1, d_{\oplus}{)}_{q}$ convolutional
code $V_{\oplus}$, where $d_{\oplus}\geq \min\{d_{1}^{\perp},
d_{2}^{\perp}\}$.
\end{theorem}
\begin{proof}
We know that $V_1$ has a generator matrix of the form
$[G(D){]}_{1} =  {\tilde H}_{0}^{(1)} + {\tilde H}_{1}^{(1)} D$,
where $$H^{(1)} = \left[
\begin{array}{c}
H_0^{(1)}\\
H_1^{(1)}\\
\end{array}
\right]$$ is a parity check matrix of $C_1$ and $\operatorname{rk}
H_0^{(1)} = k_1 \geq \operatorname{rk} H_1^{(1)} =
\operatorname{rk} {\tilde H}_1^{(1)} = {\delta}_1$. The code
$V_{2}$ has a generator matrix of the form $[G(D){]}_{2} = {\tilde
H}_{0}^{(2)} + {\tilde H}_{1}^{(2)} D$, where $$H^{(2)} = \left[
\begin{array}{c}
H_0^{(2)}\\
H_1^{(2)}\\
\end{array}
\right]$$ is a parity check matrix of $C_{2}$ and
$\operatorname{rk} H_0^{(2)} = k_{2} \geq \operatorname{rk}
H_1^{(2)} = \operatorname{rk} {\tilde H}_1^{(2)} = {\delta}_{2}$.
Consider next the code $C_{\oplus}^{\perp}$ generated by
$$H_{\oplus} =\left[
\begin{array}{cc}
H^{(1)} & 0\\
0 & H^{(2)}\\
\end{array}
\right].$$ The minimum distance of $C_{\oplus}^{\perp}$ satisfies
$d(C_{\oplus}^{\perp})=\min\{d_{1}^{\perp}, d_{2}^{\perp}\}$. We
can construct a convolutional code $V_{\oplus}$ with generator
matrix $$[G(D){]}_{\oplus} =\left[
\begin{array}{cc}
H_{0}^{(1)} & 0\\
0 & H_{0}^{(2)}\\
\end{array}
\right]+ \left[
\begin{array}{cc}
{\tilde H}_{1}^{(1)} & 0\\
0 & {\tilde H}_{1}^{(2)}\\
\end{array}
\right]D.$$ From construction and from Theorem~\ref{A}
$[G(D){]}_{\oplus}$ is a reduced basic generator matrix. Again by
construction, $V_{\oplus}$ has parameters $(n_1 + n_{2}, k_1 +
k_{2}, {\delta}_{1}+ {\delta}_{2} ; 1, d_{\oplus}{)}_{q}$ and,
from Theorem~\ref{A}, $d_{\oplus}\geq \min\{d_{1}^{\perp},
d_{2}^{\perp}\}$. Therefore, the proof is complete.
\end{proof}

\subsection{Puncturing Codes}\label{sub3.3}

The technique of puncturing codes is well-known in the literature
\cite{Macwilliams:1977,Huffman:2003}. If $C$ be an ${[n, k,
d]}_{q}$ code then we denote by $C^{i}$ the punctured code in the
coordinate $i$. Recall that the dual of a punctured code is a
shortened code. Let us show the first result of this subsection:

\begin{theorem}\label{MAINIII}
Assume that there exists an $(n, k, \delta ; 1, d_{f} {)}_{q}$
convolutional code $V$ derived from an ${[n, k^{*}, d ]}_{q}$ code
$C$, where $d^{\perp}=d(C^{\perp})> 1$ and $d_{f}\geq d^{\perp}$.
Then the following hold:
\begin{enumerate}
\item [ (i)] If $C^{\perp}$ has a minimum weight codeword with a
nonzero $i$th coordinate then there exists an $(n-1, k, \delta ;
1, d_{f_{i}} {)}_{q}$ convolutional code $V_i$, where $d_{f_{i}}
\geq d^{\perp} -1$;

\item [ (ii)] If $C^{\perp}$ has no minimum weight codeword with a
nonzero $i$th coordinate, then there exists an $(n-1, k, \delta ;
1, d_{f_{i}} {)}_{q}$ convolutional code $V_i$, where $d_{f_{i}}
\geq d^{\perp}$.
\end{enumerate}
\end{theorem}
\begin{proof}
We only show item (ii) since item (i) is similar. Let us consider
$$H = \left[
\begin{array}{c}
H_0\\
H_1\\
\end{array}
\right]$$ be a parity check matrix of $C$. Then we know that a
reduced basic generator matrix for $V$ is given by $G(D) = {H}_0 +
{\tilde H}_1 D$, where $\operatorname{rk} H_0 = k \geq
\operatorname{rk} {H}_1 = \operatorname{rk} {\tilde H}_1 =
\delta$. Since $d^{\perp} > 1$ and $C^{\perp}$ has no minimum
weight codeword with a nonzero $i$th coordinate, puncturing
$C^{\perp}$ in the $i$th coordinate one obtains the code
$C_{i}^{\perp}$ with parameters ${[n-1, n-k^{*}, d^{\perp}
]}_{q}$, where $n-k^{*}=\operatorname{rk} H_0 + \operatorname{rk}
H_1$. Consider next the convolutional code $V_i$ with generator
matrix $$[G(D){]}_{i} = {H}_0^{(i)} + {\tilde H}_1^{(i)} D,$$
where ${H}_0^{(i)}$ is derived from $H_0$ by removing column $i$
(and omitting possibly zero or duplicate rows), and ${\tilde
H}_1^{(i)}$ is derived from $H_1$ by removing column $i$ (and
omitting possibly zero or duplicate rows) and adding zero rows if
necessary (according to Theorem~\ref{A}). Since the code
$C_0^{\perp}$ generated by the matrix $H_0$ also has minimum
distance $d(C_{0}^{\perp}) > 1$ then $\operatorname{rk} H_0^{(i)}
= \operatorname{rk} H_0 = k$. By the same reasoning,
$\operatorname{rk} H_1^{(i)} = \operatorname{rk} H_1 = \delta$.
From construction and from Theorem~\ref{A}, ${[G(D)]}_{i}$ is a
reduced basic generator matrix; again by construction, the
convolutional code $V_i$ has parameters $(n-1, k, \delta ; 1,
d_{f_{i}} {)}_{q}$. Applying Theorem~\ref{A}, it follows that
$d_{f_{i}}\geq d^{\perp}$, so we get the result. Note that since
$(C_{i}^{\perp})={[C_{S_{i}}]}^{\perp}$ (where the latter code is
the dual of the shortened code), the code $C_{S_{i}}$ constructed
in this proof plays the role of the code $C$ given in
Theorem~\ref{A}.
\end{proof}

\begin{remark}
Note that the procedure adopted in Theorem~\ref{MAINIII} can be
generalized to puncture codes on two or more coordinates.
\end{remark}

Let us see how to obtain convolutional codes by puncturing
memory-two and memory-three convolutional codes:

\begin{theorem}\label{MAINIIIP}
Assume that there exists an $(n, k, \delta ; 2, d_{f} {)}_{q}$
convolutional code $V$ derived from an ${[n, k^{*}, d ]}_{q}$ code
$C$, where $d^{\perp}=d(C^{\perp})> 1$ and $d_{f}\geq d^{\perp}$.
Then the following hold:
\begin{enumerate}
\item [ (i)] If $C^{\perp}$ has a minimum weight codeword with a
nonzero $i$th coordinate then there exists an $(n-1, k, \delta ;
2, d_{f_{i}} {)}_{q}$ convolutional code $V_i$, where $d_{f_{i}}
\geq d^{\perp} -1$;

\item [ (ii)] If $C^{\perp}$ has no minimum weight codeword with a
nonzero $i$th coordinate, then there exists an $(n-1, k, \delta ;
2, d_{f_{i}} {)}_{q}$ convolutional code $V_i$, where $d_{f_{i}}
\geq d^{\perp}$.
\end{enumerate}
\end{theorem}
\begin{proof}
We only prove item (ii); item (i) follows similarly. Let us
consider
$$H = \left[
\begin{array}{c}
H_0\\
H_1\\
H_{2}\\
\end{array}
\right]$$ be a parity check matrix of $C$. Then a reduced basic
generator matrix for $V$ is given by $$G(D) = {H}_0 + {\tilde H}_1
D + {\tilde H}_{2} D^{2}.$$ We have two cases to be considered.
The first one is $\operatorname{rk} H_0 \geq \operatorname{rk}
{H}_{2} \geq \operatorname{rk} {H}_1$. In this case it follows
that $\operatorname{rk} H_0 = k$ and $\delta = 2\operatorname{rk}
{H}_{2}$ and the computation of the parameters of the final code
is performed without pain.

Let us investigate the second case, i.e., $\operatorname{rk} H_0
\geq \operatorname{rk} {H}_{1} \geq \operatorname{rk} {H}_{2}$. In
the second case we have $\operatorname{rk} H_0 = k$ and $\delta =
\operatorname{rk} {H}_{2}+\operatorname{rk} {H}_{1}$. Since
$d^{\perp} > 1$ and $C^{\perp}$ has no minimum weight codeword
with a nonzero $i$th coordinate, puncturing $C^{\perp}$ in the
$i$th coordinate one obtains the code $C_{i}^{\perp}$ with
parameters ${[n-1, n-k^{*}, d^{\perp} ]}_{q}$, where
$n-k^{*}=\operatorname{rk} H_0 + \operatorname{rk} H_1 +
\operatorname{rk} H_{2}$. Let us consider the code $V_i$ with
generator matrix
$$[G(D){]}_{i} = {H}_0^{(i)} + {\tilde H}_1^{(i)} D + + {\tilde H}_{2}^{(i)} D^{2},$$ where
${H}_0^{(i)}$, ${H}_1^{(i)}$ and ${H}_{2}^{(i)}$ are derived from
$H_0$, $H_1$, $H_{2}$, respectively, by removing column $i$ (and
omitting possibly zero or duplicate rows). Since the code
$C_0^{\perp}$ generated by the matrix $H_0$ also has minimum
distance $d(C_{0}^{\perp}) > 1$ then $\operatorname{rk} H_0^{(i)}
= \operatorname{rk} H_0 = k$. By the same reasoning it follows
that $\operatorname{rk} H_1^{(i)} = \operatorname{rk} H_1$ and
$\operatorname{rk} H_{2}^{(i)} = \operatorname{rk} H_{2}$. From
construction and from Theorem~\ref{A}, ${[G(D)]}_{i}$ is a reduced
basic generator matrix of a memory-two convolutional code $V_i$ of
dimension $k$. Since $\operatorname{rk} H_1^{(i)}\geq
\operatorname{rk} H_{2}^{(i)}$ it follows that the degree
${\delta}^{(i)}$ of $V_{i}$ is given by ${\delta}^{(i)} =
2\operatorname{rk} {H}_{2}^{(i)}+(\operatorname{rk}
{H}_{1}^{(i)}-\operatorname{rk} {H}_{2}^{(i)})= \delta$. Therefore
the convolutional code $V_i$ has parameters $(n-1, k, \delta ; 2,
d_{f_{i}} {)}_{q}$. Applying Theorem~\ref{A}, it follows that
$d_{f_{i}}\geq d^{\perp}$, and the result follows.
\end{proof}

\begin{theorem}\label{MAINIIIPA}
Assume that there exists an $(n, k, \delta ; 3, d_{f} {)}_{q}$
convolutional code $V$ derived from an ${[n, k^{*}, d ]}_{q}$ code
$C$, where $d^{\perp}=d(C^{\perp})> 1$ and $d_{f}\geq d^{\perp}$.
Then the following hold:
\begin{enumerate}
\item [ (i)] If $C^{\perp}$ has a minimum weight codeword with a
nonzero $i$th coordinate then there exists an $(n-1, k, \delta ;
3, d_{f_{i}} {)}_{q}$ convolutional code $V_i$, where $d_{f_{i}}
\geq d^{\perp} -1$;

\item [ (ii)] If $C^{\perp}$ has no minimum weight codeword with a
nonzero $i$th coordinate, then there exists an $(n-1, k, \delta ;
3, d_{f_{i}} {)}_{q}$ convolutional code $V_i$, where $d_{f_{i}}
\geq d^{\perp}$.
\end{enumerate}
\end{theorem}
\begin{proof}
 Let us consider
$$H = \left[
\begin{array}{c}
H_0\\
H_1\\
H_{2}\\
H_3\\
\end{array}
\right]$$ be a parity check matrix of $C$. We only compute the
degree of the corresponding convolutional code $V_{i}$ since the
computation of the other parameters of $V_{i}$ are similar as in
the above proof. We have six cases. The cases $\operatorname{rk}
{H}_{3} \geq \operatorname{rk} {H}_{2} \geq \operatorname{rk}
{H}_{1}$ and $\operatorname{rk} {H}_{3} \geq \operatorname{rk}
{H}_{1} \geq \operatorname{rk} {H}_{2}$ are straightforward.

If the inequalities $\operatorname{rk} {H}_{2} \geq
\operatorname{rk} {H}_{3} \geq \operatorname{rk} {H}_{1}$ or
$\operatorname{rk} {H}_{2} \geq \operatorname{rk} {H}_{1} \geq
\operatorname{rk} {H}_{3}$ are true then it implies that $\delta
=3\operatorname{rk} {H}_{3} + 2(\operatorname{rk} {H}_{2}-
\operatorname{rk} {H}_{3})={\delta}^{(i)}$.

If the inequalities $\operatorname{rk} {H}_{1} \geq
\operatorname{rk} {H}_{3} \geq \operatorname{rk} {H}_{2}$ hold
then it follows that $\delta =3\operatorname{rk} {H}_{3} +
(\operatorname{rk} {H}_{1}- \operatorname{rk}
{H}_{3})={\delta}^{(i)}$.

Finally, if the inequalities $\operatorname{rk} {H}_{1} \geq
\operatorname{rk} {H}_{2} \geq \operatorname{rk} {H}_{3}$ are true
then one has $\delta =3\operatorname{rk} {H}_{3} +
2(\operatorname{rk} {H}_{2}- \operatorname{rk} {H}_{3}) +
(\operatorname{rk} {H}_{1}- \operatorname{rk} {H}_{2})
={\delta}^{(i)}$.
\end{proof}

\subsection{Code Extension}\label{sub3.4}

As in the case of puncturing codes, the code extension technique
\cite{Macwilliams:1977,Huffman:2003} is well known in the
literature. Here we extend this technique to convolutional codes.

Let $C$ be an ${[n, k, d]}_{q}$ linear code over ${\mathbb
F}_{q}$. The extended code $C^{e}$ is the linear code given by
$C^{e} = \{ (x_1,\ldots, x_{n}, x_{n+1}) \in {\mathbb F}_{q}^{n+1}
|$ $(x_1,\ldots, x_{n})\in C, x_1 +\cdots +x_{n}+ x_{n+1}=0 \}$.
The code $C^{e}$ is linear and has parameters ${[n+1, k,
d^{e}]}_{q}$, where $d^{e}=d$ or $d^{e}=d+1$. Recall that a vector
${ \bf v}=(v_1, \ldots, v_{n}) \in F_{q}^{n}$ is called
\emph{even-like} if it satisfies the equality
$\displaystyle\sum_{i=1}^{n} v_{i}=0$, and \emph{odd-like}
otherwise. For an ${[n, k, d]}_{q}$ code $C$, the minimum weight
of the even-like codewords of $C$ is called \emph{minimum
even-like weight} and it is denoted by $d_{even}$ (or
${(d)}_{even}$). Similarly, the minimum weight of the odd-like
codewords of $C$ is called \emph{minimum odd-like weight} and it
is denoted by $d_{odd}$ (or ${(d)}_{odd}$).

Theorem~\ref{MAINIV} is the main result of this subsection:

\begin{theorem}\label{MAINIV}
Suppose that an $(n, k, \delta ; 1, d_{f} {)}_{q^{m}}$
convolutional code $V$ derived from an ${[n, k^{*},
d^{\perp}]}_{q}$ code $C^{\perp}$ with $d_{f}\geq d$ exists. Then
the following hold:
\begin{itemize}
\item[ (i)] If ${(d)}_{even}\leq {(d)}_{odd}$, then there exists
an $(n+1, k, \delta ; 1, d_{f^{e}}{)}_{q}$ convolutional code,
where $d_{f^{e}} \geq d$;

\item[ (ii)] If ${(d)}_{odd} < {(d)}_{even}$, then there exists an
$(n+1, k, \delta ; 1, d_{f^{e}} {)}_{q}$ convolutional code, where
$d_{f^{e}} \geq d+1$
\end{itemize}
\end{theorem}

\begin{proof}
We only show item (ii). The code $V$ admits a reduced basic
generator matrix of the form $G(D) = {G}_0 + {\tilde G}_1 D$,
where $$G = \left[
\begin{array}{c}
G_0\\
G_1\\
\end{array}
\right]$$ is a generator matrix of $C$ with $\operatorname{rk} G_0
= k \geq \operatorname{rk} {G}_1 = \operatorname{rk} {\tilde G}_1
= \delta$. Consider the extended code $C^{e}$ generated by $$G^{e}
= \left[
\begin{array}{c}
G_{0}^{e}\\
G_{1}^{e}\\
\end{array}
\right].$$ From construction and from Theorem~\ref{A}, the matrix
${[G(D)]}^{e} = {G}_{0}^{e} + {\tilde G}_{1}^{e} D$ is a reduced
and basic generator matrix of a convolutional code $V^{e}$. By
construction, $V^{e}$ has parameters $(n+1, k, \delta ; 1,
d_{f^{e}} {)}_{q}$. From Theorem~\ref{A} and from hypothesis,
$d_{f^{e}}\geq d+1 $; hence the result follows.
\end{proof}

\subsection{The $({ \bf u}| { \bf u}+{ \bf v})$ Construction}\label{sub3.5}

The $({ \bf u}| { \bf u}+{ \bf v})$ construction
\cite{Macwilliams:1977,Huffman:2003} is another method for
constructing new classical linear codes.

Let $C_1$ and $C_{2}$ be two linear codes of same length both over
${\mathbb F}_{q}$ with parameters ${[n, k_1, d_1]}_{q}$ and ${[n,
k_{2}, d_{2}]}_{q}$, respectively. Then by applying the $({ \bf
u}| { \bf u}+{ \bf v})$ construction one can generate a new code
$C = \{ ({ \bf u}, { \bf u}+{ \bf v}) | { \bf u} \in C_1 , { \bf
v} \in C_{2} \}$ with parameters ${[2n, k_1 + k_{2}, \min \{
2d_{1}, d_{2} \}]}_{q}$.

Theorem~\ref{MAINV} is one of the main results of this subsection:

\begin{theorem}\label{MAINV}
Assume that there exist an $(n, k_1, {\delta}_1 ; 1, d_{f}^{(1)}
{)}_{q}$ convolutional code $V_1$ derived from an ${[n, k_1^{*},
d_1]}_{q}$ code $C_1$, where $d_{f}^{(1)} \geq d_{1}^{\perp}$, and
an $(n, k_{2}, {\delta}_{2} ; 1, d_{f}^{(2)} {)}_{q}$
convolutional code $V_{2}$ derived from an ${[n, k_{2}^{*},
d_{2}]}_{q}$ code $C_{2}$, where $d_{f}^{(2)} \geq d_{2}^{\perp}$.
Then there exists an $(2n, k_1 + k_{2}, {\delta}_{1}+ {\delta}_{2}
; 1, d_{f_{({ \bf u}| { \bf u}+{ \bf v})}} {)}_{q}$ convolutional
code $V_{({ \bf u}| { \bf u}+{ \bf v})}$, where $d_{f_{({ \bf u}|
{ \bf u}+{ \bf v})}}\geq \min\{2d_{2}^{\perp}, d_{1}^{\perp}\}$.
\end{theorem}

\begin{proof}
The code $V_1$ has a (reduced basic) generator matrix of the form
${[G(D)]}_{1} = {H}_0^{(1)} + {\tilde H}_1^{(1)} D$, where
$$H^{(1)} = \left[
\begin{array}{c}
H_0^{(1)}\\
H_1^{(1)}\\
\end{array}
\right]$$ is a parity check matrix of $C_1$ and $\operatorname{rk}
H_0^{(1)} = k_1 \geq \operatorname{rk} {H}_1^{(1)} =
\operatorname{rk} {\tilde H}_1^{(1)} = {\delta}_{1}$. The code
$V_{2}$ has a (reduced basic) generator matrix of the form
${[G(D)]}_{2} = {H}_0^{(2)} + {\tilde H}_1^{(2)} D$, where
$$H^{(2)} = \left[
\begin{array}{c}
H_0^{(2)}\\
H_1^{(2)}\\
\end{array}
\right]$$ is a parity check matrix of $C_{2}$ and
$\operatorname{rk} H_0^{(2)} = k_{2} \geq \operatorname{rk}
{H}_1^{(2)} = \operatorname{rk} {\tilde H}_1^{(2)} =
{\delta}_{2}$. Consider the code $C_{({ \bf u}| { \bf u}+{ \bf
v})}^{\perp}$ generated by the matrix $$H_{({ \bf u}| { \bf u}+{
\bf v})} = \left[
\begin{array}{cc}
H^{(1)} & 0\\
-H^{(2)} & H^{(2)}\\
\end{array}
\right].$$ Let us compute the minimum distance $ d( C_{({ \bf u}|
{ \bf u}+{ \bf v})}^{\perp})$ of $C_{({ \bf u}| { \bf u}+{ \bf
v})}^{\perp}$. The codewords of $C_{({ \bf u}| { \bf u}+{ \bf
v})}^{\perp}$ are of the form $\{ ({ \bf u} -{ \bf v}, { \bf v})|
{ \bf u} \in C_{1}^{\perp}, { \bf v} \in C_{2}^{\perp} \}$. Let us
consider the the codeword ${ \bf w}=({ \bf u} -{ \bf v}, { \bf
v})$. If ${ \bf u}= { \bf 0}$ then ${ \bf w}=(-{ \bf v}, { \bf
v})$, so the minimum distance of $C_{({ \bf u}| { \bf u}+{ \bf
v})}^{\perp}$ is given by $2d_{2}^{\perp}$. On the other hand, if
${ \bf u} \neq { \bf 0}$ then wt$({ \bf w})=$wt$({ \bf u} -{ \bf
v})+$wt$({ \bf v})=$d$({ \bf u}, { \bf v})+$d$(({ \bf v}, { \bf
0})\geq$d$({ \bf u}, { \bf 0})=$wt$({ \bf u})$. Thus, in this
case, the minimum distance of $C_{({ \bf u}| { \bf u}+{ \bf
v})}^{\perp}$ is given by $d_{1}^{\perp}$. Therefore, $ d( C_{({
\bf u}| { \bf u}+{ \bf v})}^{\perp})= \min \{ 2d_{2}^{\perp},
d_{1}^{\perp}\}$. Next, we can construct a new convolutional code
$V_{({ \bf u}| { \bf u}+{ \bf v})}$ with generator matrix
$${[G(D)]}_{({ \bf u}| { \bf u}+{ \bf v})} = \left[
\begin{array}{cc}
H_{0}^{(1)} & 0\\
-H_{0}^{(2)} & H_{0}^{(2)}\\
\end{array}
\right]+
\left[
\begin{array}{cc}
{\tilde H}_{1}^{(1)} & 0\\
-{\tilde H}_{1}^{(2)} & {\tilde H}_{1}^{(2)}\\
\end{array}
\right] D.$$ From construction and from Theorem~\ref{A},
${[G(D)]}_{({ \bf u}| { \bf u}+{ \bf v})}$ is a reduced basic
generator matrix for $V_{({ \bf u}| { \bf u}+{ \bf v})}$. Again
from construction, the code $V_{({ \bf u}| { \bf u}+{ \bf v})}$
has parameters $(2n, k_1 + k_{2}, {\delta}_{1}+ {\delta}_{2} ; 1,
d_{f_{({ \bf u}| { \bf u}+{ \bf v})}} {)}_{q}$. Finally, from
Theorem~\ref{A}, it follows that $d_{f_{({ \bf u}| { \bf u}+{ \bf
v})}} \geq \min \{ 2d_{2}^{\perp}, d_{1}^{\perp}\}$, and the
result follows.
\end{proof}

Let us see how to construct memory-two convolutional codes derived
from the $({ \bf u}| { \bf u}+{ \bf v})$ construction. We call the
attention for the fact that all the proposed methods can be
generalized for multi-memory codes.

\begin{theorem}\label{MAINVI}
Assume that there exist an $(n, k_1, {\delta}_1 ; 2, d_{f}^{(1)}
{)}_{q}$ convolutional code $V_1$ derived from an ${[n, k_1^{*},
d_1]}_{q}$ code $C_1$, where $d_{f}^{(1)} \geq d_{1}^{\perp}$, and
an $(n, k_{2}, {\delta}_{2} ; 2, d_{f}^{(2)} {)}_{q}$
convolutional code $V_{2}$ derived from an ${[n, k_{2}^{*},
d_{2}]}_{q}$ code $C_{2}$, where $d_{f}^{(2)} \geq d_{2}^{\perp}$.
Then there exists an $(2n, k_1 + k_{2}, {\delta}_{1}+ {\delta}_{2}
; 2, d_{f_{({ \bf u}| { \bf u}+{ \bf v})}} {)}_{q}$ convolutional
code $V_{({ \bf u}| { \bf u}+{ \bf v})}$, where $d_{f_{({ \bf u}|
{ \bf u}+{ \bf v})}}\geq \min\{2d_{2}^{\perp}, d_{1}^{\perp}\}$.
\end{theorem}
\begin{proof}
The code $V_1$ has a (reduced basic) generator matrix of the form
${[G(D)]}_{1} = {H}_0^{(1)} + {\tilde H}_1^{(1)} D + {\tilde
H}_{2}^{(1)} D^{2}$, where $$H^{(1)} = \left[
\begin{array}{c}
H_0^{(1)}\\
H_1^{(1)}\\
H_{2}^{(1)}\\
\end{array}
\right]$$ is a parity check matrix of $C_1$. We have to consider
two cases. The first case is when $\operatorname{rk} H_0^{(1)} =
k_1 \geq \operatorname{rk} {H}_{2}^{(1)} \geq \operatorname{rk}
{H}_{1}^{(1)}$. In this case the computation of the parameters of
the final code is straightforward. Let us consider the second
case, i.e., $\operatorname{rk} H_0^{(1)} = k_1 \geq
\operatorname{rk} {H}_1^{(1)} \geq \operatorname{rk}
{H}_{2}^{(1)}$. We know that ${\delta}_1 =
2\operatorname{rk}H_{2}^{(1)} + (\operatorname{rk}
{H}_1^{(1)}-\operatorname{rk}H_{2}^{(1)})$. Similarly, let us
consider the code $V_{2}$ has a (reduced basic) generator matrix
of the form ${[G(D)]}_{2} = {H}_0^{(2)} + {\tilde H}_1^{(2)} D +
{\tilde H}_{2}^{(2)} D^{2}$, where $$H^{(2)} = \left[
\begin{array}{c}
H_0^{(2)}\\
H_1^{(2)}\\
H_{2}^{(2)}\\
\end{array}
\right]$$ is a parity check matrix of $C_{2}$; without loss of
generality we may assume that $\operatorname{rk} H_0^{(2)} = k_{2}
\geq \operatorname{rk} {H}_1^{(2)} \geq \operatorname{rk}
{H}_{2}^{(2)} $. We know that ${\delta}_{2} =
2\operatorname{rk}H_{2}^{(2)} + (\operatorname{rk}
{H}_1^{(2)}-\operatorname{rk}H_{2}^{(2)})$. Consider the code
$C_{({ \bf u}| { \bf u}+{ \bf v})}^{\perp}$ generated by the
matrix $$H_{({ \bf u}| { \bf u}+{ \bf v})} = \left[
\begin{array}{cc}
H^{(1)} & 0\\
-H^{(2)} & H^{(2)}\\
\end{array}
\right].$$ From the proof of Theorem~\ref{MAINV}, the minimum
distance of $C_{({ \bf u}| { \bf u}+{ \bf v})}^{\perp}$ is given
by $ d( C_{({ \bf u}| { \bf u}+{ \bf v})}^{\perp})= \min \{
2d_{2}^{\perp}, d_{1}^{\perp}\}$.

Next, we can construct a new convolutional code $V_{({ \bf u}| {
\bf u}+{ \bf v})}$ with generator matrix
\begin{eqnarray*}
{[G(D)]}_{({ \bf u}| { \bf u}+{ \bf v})} = \left[
\begin{array}{cc}
H_{0}^{(1)} & 0\\
-H_{0}^{(2)} & H_{0}^{(2)}\\
\end{array}
\right]+ \left[
\begin{array}{cc}
{\tilde H}_{1}^{(1)} & 0\\
-{\tilde H}_{1}^{(2)} & {\tilde H}_{1}^{(2)}\\
\end{array}
\right] D \\ + \left[
\begin{array}{cc}
{\tilde H}_{2}^{(1)} & 0\\
-{\tilde H}_{2}^{(2)} & {\tilde H}_{2}^{(2)}\\
\end{array}
\right] D^{2}.
\end{eqnarray*}
From construction and from Theorem~\ref{A}, ${[G(D)]}_{({ \bf u}|
{ \bf u}+{ \bf v})}$ is a reduced basic generator matrix for
$V_{({ \bf u}| { \bf u}+{ \bf v})}$. Again from construction, the
code $V_{({ \bf u}| { \bf u}+{ \bf v})}$ has parameters $(2n, k_1
+ k_{2}, {\delta}_{({ \bf u}| { \bf u}+{ \bf v})} ; 2, d_{f_{({
\bf u}| { \bf u}+{ \bf v})}} {)}_{q}$. It is not difficult to see
that ${\delta}_{({ \bf u}| { \bf u}+{ \bf v})}= {\delta}_{1}+
{\delta}_{2}$. From Theorem~\ref{A}, it follows that $d_{f_{({ \bf
u}| { \bf u}+{ \bf v})}} \geq \min \{ 2d_{2}^{\perp},
d_{1}^{\perp}\}$. Therefore one can get an $(2n, k_1 + k_{2},
{\delta}_{1}+ {\delta}_{2} ; 2, d_{f_{({ \bf u}| { \bf u}+{ \bf
v})}} {)}_{q}$ convolutional code, where $d_{f_{({ \bf u}| { \bf
u}+{ \bf v})}}\geq \min\{2d_{2}^{\perp}, d_{1}^{\perp}\}$, and the
result follows.
\end{proof}

\subsection{Product codes}\label{sub3.6}

In this section we show how to construct convolutional codes by
applying the product code construction.

Recall the given by classical linear codes $C_1 = {[n_1, k_1,
d_1]}_{q}$ and $C_{2}={[n_{2}, k_{2}, d_{2}]}_{q}$ over ${\mathbb
F}_{q}$, the direct product code denoted by $C_1 \otimes C_{2}$ is
the linear code over ${\mathbb F}_{q}$ with parameters $C_1 =
{[n_{1}n_{2}, k_{1}k_{2}, d_{1}d_{2}]}_{q}$. If
$G_1=(g_{ij}^{(1)})$ and $G_{2}=(g_{ij}^{(2)})$ are generator
matrices for $C_1$ and $C_{2}$ respectively, then the Kronecker
product $G_1 \otimes G_{2}=(g_{ij}^{(1)}G_{2})$ is a generator
matrix for $C_1 \otimes C_{2}$. We are now ready to show the next
result:

\begin{theorem}\label{MAINproducode}
Assume that there exist an $(n_1, k_1, {\delta}_1 ; 1, d_{f}^{(1)}
{)}_{q}$ convolutional code $V_1$ derived from an ${[n_{1},
k_{1}^{*}, d_{1}^{\perp} ]}_{q}$ code $C_1^{\perp}$, where
$d_{f}^{(1)}\geq d_1$, and an $(n_{2}, k_{2}, {\delta}_{2} ; 1,
d_{f}^{(2)} {)}_{q}$ convolutional code $V_{2}$ derived from an
${[n_{2}, k_{2}, d_{2}^{\perp}]}_{q}$ code $C_{2}^{\perp}$, where
$d_{f}^{(2)}\geq d_{2}$. Then there exists an $(n_{1} n_{2},
k_{1}k_{2}, {\delta}_{1}{\delta}_{2} ; 1, d_{\otimes}{)}_{q}$
convolutional code $V_{\otimes}$, where $d_{\otimes}\geq
d_{1}d_{2}$.
\end{theorem}
\begin{proof}
The code $V_1$ has a generator matrix of the form $[G(D){]}_{1} =
{\tilde G}_{0}^{(1)} + {\tilde G}_{1}^{(1)} D$, where $$G^{(1)} =
\left[
\begin{array}{c}
G_0^{(1)}\\
G_1^{(1)}\\
\end{array}
\right]$$ is a generator matrix of the code $C_1={[n_{1}, n_1 -
k_{1}^{*}, d_1 ]}_{q}$ and $\operatorname{rk} G_0^{(1)} = k_1 \geq
\operatorname{rk} G_1^{(1)} = \operatorname{rk} {\tilde G}_1^{(1)}
= {\delta}_1$. The code $V_{2}$ has a generator matrix of the form
$[G(D){]}_{2} = {\tilde G}_{0}^{(2)} + {\tilde G}_{1}^{(2)} D$,
where $$G^{(2)} = \left[
\begin{array}{c}
G_0^{(2)}\\
G_1^{(2)}\\
\end{array}
\right]$$ is a generator matrix of the code $C_{2}={[n_{2}, n_{2}
- k_{2}^{*}, d_{2} ]}_{q}$ and $\operatorname{rk} G_0^{(2)} =
k_{2} \geq \operatorname{rk} G_1^{(2)} = \operatorname{rk} {\tilde
G}_1^{(2)} = {\delta}_{2}$. Consider next the (linear) code
$C_{\otimes}$ generated by the matrix $$M_{\otimes} = \left[
\begin{array}{c}
G_0^{(1)}\otimes G_0^{(2)}\\
G_1^{(1)}\otimes G_1^{(2)}\\
\end{array}
\right].$$ We know that $C_{\otimes}$ is a subcode of
$G^{(1)}\otimes G^{(2)}$, so $C_{\otimes}$ has minimum distance
greater than or equal to $d_1 d_{2}$. We construct a convolutional
code $V_{\otimes}$ with generator matrix
$$[G(D){]}_{\otimes} = [G_{0}^{(1)} \otimes G_{0}^{(2)}]+ { \tilde G}D,$$ where $G
=[{G}_{1}^{(1)}\otimes {G}_{1}^{(2)}]$. From construction and from
Theorem~\ref{A} $[G(D){]}_{\otimes}$ is a reduced basic generator
matrix. Again by construction, $V_{\otimes}$ has parameters
$(n_{1} n_{2}, k_{1} k_{2}, {\delta}_{1} {\delta}_{2} ; 1,
d_{\otimes}{)}_{q}$, and from Theorem~\ref{A}, $d_{\otimes}\geq
d_{1}d_{2}$. Thus the result follows.
\end{proof}

\section{Code Constructions and Code Tables}\label{IV}

In this section we utilize the convolutional code expansion
developed in Subsection~\ref{sub3.1} to obtain new families of
convolutional codes. In order to shorten the length of this paper
we only apply this technique, although it is clear that all
construction methods proposed in Section~\ref{III} can also be
applied in order to construct new families of convolutional codes.
Moreover, we do not construct explicitly the new families of
convolutional codes shown in Tables~\ref{table1} and \ref{table2}
because the constructions are simple applications of
Theorem~\ref{MAINI}.

\begin{remark}
It is important to observe that in all convolutional code
presented in Tables~\ref{table1} and \ref{table2}, we expand the
codes defined over ${\mathbb F}_{q}$ (where $q=p^{t}$, $t\geq 1$
and $p$ prime) with respect to the prime field ${\mathbb F}_{p}$.
However, the method also holds if one expands such a codes over
any subfield of the field ${\mathbb F}_{q}$.
\end{remark}

In Tables~\ref{table1} and \ref{table2} we tabulated only
convolutional BCH codes shown in \cite{LaGuardia:2012} for the
Hermitian case, although in \cite{LaGuardia:2012} we also
construct families of convolutional BCH codes for the Euclidean
case.

\section{Summary}\label{V}

We have shown how to construct new families of convolutional codes
by applying the techniques of puncturing, extending, expanding,
direct sum, the $({ \bf u}| { \bf u}+{ \bf v})$ construction and
the product code construction. As examples of code expansion, new
convolutional codes derived from cyclic and character codes were
constructed.

\section*{Acknowledgment}
This research has been partially supported by the Brazilian
Agencies CAPES and CNPq.

\clearpage

\begin{table}[bth!]
\begin{center}
\caption{Families of Convolutional Codes \label{table1}}
\begin{tabular}{|c| c| c|}
\hline Code Family / $(n, k, \delta; \mu, d_{f}{)}_{q}$ & Range of Parameters & Ref.\\
\hline BCH & &\\
\hline $(2^{2m}-2^{m}, 2^{2m}-2^{m}-m, m; \mu, 3{)}_{2}$ & $m\geq 1$ & \cite{Hole:2000}\\
\hline $(2^{2m-1}, 2^{2m-1}-m, m; \mu, 4{)}_{2}$ & $m\geq 2$ & \cite{Hole:2000}\\
\hline $(2^{m}-1, 2^{m}-1-(t-1)m, m; \mu, d_{f}\geq 2t+1{)}_{2}$ & $2 \leq t < 2^{\lceil m/2 \rceil -1} +1, m\geq 3$ & \cite{Hole:2000}\\
\hline $(2^{m}-1, 2^{m}-2-(t-1)m, m; \mu, d_{f}\geq 2t+2{)}_{2}$ & $2 \leq t < 2^{\lceil m/2 \rceil -1} +1, m\geq 3$ & \cite{Hole:2000}\\
\hline $(n, n-r\lceil \delta(1-1/q)\rceil, \delta ; 1, d_{f}\geq
\delta +1 + \Delta(\delta+1, 2\delta) {)}_{q}$ &
& \cite{Aly:2007}\\
\hline  & $\gcd(n, q)=1$, $r=$ord$_{n}(q)$, $2\leq 2\delta < {\delta}_{max}$ &\\
\hline  & ${\delta}_{max}= \lfloor n/(q^{r}-1) (q^{\lceil r/2 \rceil} -1-(q-2)[r odd])\rfloor$ & \\
\hline  & $\Delta(\alpha, \beta)\geq q + \lfloor (\beta - \alpha +3)/q \rfloor-2$, if $\beta - \alpha \geq 2q-3$, & \\
\hline  & $\Delta(\alpha, \beta)\geq \lfloor (\beta - \alpha +3)/2 \rfloor$, otherwise & \\
\hline $(n, n-2(i-2)-1, 2; \mu, d_{f}\geq i+1)_{q^2}$ &  &\cite{LaGuardia:2012}\\
\hline  & $n = q^{4} - 1$, $q\geq 3$ prime power, & \\
\hline & $3\leq i\leq q^{2}-1$ &\\
\hline $(n, n-m(2q^2-3)-1, m; \mu , d_{f}\geq 2q^2 +
2)_{q^2}$  &  &\cite{LaGuardia:2012}\\
\hline  &  $n=q^{2m}-1$, $q\geq 4$ prime power, &\\
\hline  & $m= \ {{ord}_{n}}(q^2)\geq 3$ &\\
\hline $(n, n-mi-1, mj; \mu, d_{f}\geq i+j+2
)_{q^2}$ & &\cite{LaGuardia:2012}\\
\hline  & $n=q^{2m}-1$, $q\geq 4$ prime power, &\\
\hline  & $m= \
{{ord}_{n}}(q^2)\geq 3$, $1\leq i=j\leq q^2 - 2$ &\\
\hline $(n, n-m(i-1)-1, m; \mu, d_{f}\geq i+2)_{q^2}$ &  &\cite{LaGuardia:2012}\\
\hline  & $n=q^{2m}-1$, $q\geq 4$ prime power, &\\
\hline  &$m= \ {{ord}_{n}}(q^2)\geq 3$, $1\leq i < q^2 - 1$&\\
\hline $(n, n-m(i-2)-1, 2m; \mu, d_f \geq i+2)_{q^2}$ & &\cite{LaGuardia:2012}\\
\hline & $n=q^{2m}-1$, $q\geq 4$ prime power, & \\
\hline & $m= \ {{ord}_{n}}(q^2)\geq 3$, $3\leq i < q^2 - 1$&\\
\hline $(n, n-m(i-\mu)-1, m\mu; {\mu}^{*}, d_{f}
\geq i-\mu + 4)_{q^2}$ & &\cite{LaGuardia:2012}\\
\hline & $n=q^{2m}-1$, $q\geq 4$ prime power, & \\
\hline & $m= \ {{ord}_{n}}(q^2) \geq 3$, $\mu \geq 3$, $\mu +
1\leq i < q^2 - 1$ &\\
\hline $(n, n-2i-1, 2j; \mu, d_{f}\geq i+j+2)_{q^2}$ &  &\cite{LaGuardia:2012}\\
\hline & $n = q^{4} - 1$, $q\geq 3$ prime power,  &\\
\hline & $1\leq i=j$, $2\leq i+j\leq q^{2}-2$ &\\
\hline
\hline Expanded BCH & &\\
\hline Apply Theorem~\ref{MAINI} in the BCH codes above &  & New\\
\hline
\hline RS & &\\
\hline $(n, {\mu}/2, {\mu}/2; 1, d_{f} > \mu +1 {)}_{q}$ &  & \cite{Klappi:2007}\\
\hline  & $q$ prime power, $n$ odd divisor of $q^{2}-1$ & \\
\hline  & $q+1 < n < q^{2}-1$, $2\leq \mu=2t\leq \lfloor n/(q+1)\rfloor$ & \\
\hline
\hline Expanded RS & &\\
\hline $(rn, r{\mu}/2, r{\mu}/2; 1, d_{f} > \mu +1 {)}_{p}$ &  & New\\
\hline  & $q=p^{r}$ prime power, $n$ odd divisor of $q^{2}-1$ & \\
\hline  & $q+1 < n < q^{2}-1$, $2\leq \mu=2t\leq \lfloor n/(q+1)\rfloor$ & \\
\hline
\hline RM & &\\
\hline $(2^{m-l}, k(r), \delta \leq 2^{l}-1; \mu, d_{f}=2^{m-r} {)}_{2}$ &  & \cite{Klappi:2007}\\
\hline  & $1\leq l \leq m$, $k(r)=\displaystyle\sum_{i=0}^{r}
\left(
\begin{array}{c}
m-l\\
i\\
\end{array}
\right)$, &\\
\hline & $0\leq r\leq \lfloor (m-l-1)/2 \rfloor$ &\\
\hline
\hline MDS & &\\
\hline ${(n, 1, \delta; \mu, d_{f}=n(\delta +1))}_{q}$  & $0\leq \delta\leq n-1, \ n\leq q -1$ & \cite{Langfeld:2005}\\
\hline ${(n, n-2i, 2; 1, 2i+3)}_{q}$ & & \cite{LaGuardia:2013P}\\
\hline & $1\leq i \leq \frac{q}{2} - 2$, $q=2^t$, &\\
\hline & $t\geq 3$ integer, $n=q+1$ &\\
\hline ${(n, n-2i+1, 2; 1, 2i+2)}_{q}$ & & \cite{LaGuardia:2013P}\\
\hline & $q=p^{t}$, $t\geq 2$ integer, $p$ odd prime, &\\
\hline & $n=q+1$, $1\leq i \leq \frac{n}{2} - 2$ &\\
\hline
\hline Expanded MDS & &\\
\hline ${(tn, t, t\delta; \mu, d_{f}\geq n(\delta +1))}_{p}$  & & New\\
\hline & $q=p^{t}$ prime power &\\
\hline & $0\leq \delta\leq n-1, \ n\leq q -1$ &\\
\hline ${(tn, t(n-2i), 2t; 1, d_f \geq 2i+3)}_{2}$ & & New\\
\hline & $1\leq i \leq \frac{q}{2} - 2$, $q=2^t$, &\\
\hline & $t\geq 3$ integer, $n=q+1$ &\\
\hline ${(tn, t(n-2i+1), 2t; 1, d_f \geq 2i+2)}_{p}$ & & New\\
\hline & $q=p^{t}$, $t\geq 2$ integer, $p$ odd prime, &\\
\hline & $n=q+1$, $1\leq i \leq \frac{n}{2} - 2$ &\\
\hline
\end{tabular}
\end{center}
\end{table}

\clearpage

\begin{table}[bth!]
\begin{center}
\caption{Families of Convolutional Codes \label{table2}}
\begin{tabular}{|c| c| c|}
\hline Code Family / $(n, k, \delta; \mu, d_{f}{)}_{q}$ & Range of Parameters & Ref.\\
\hline Character & &\\
\hline $(2^{m}, 2^{m}- s_{m}(u), s_{m}(u)- s_{m}(r) ; 1, d_{f}\geq
2^{r+1} {)}_{q}$ & & \cite{LaGuardia:2013B}\\
\hline $(2^{m}, s_{m}(u), s_{m}(u)- s_{m}(r) ; \mu,
d_{f}^{\perp}\geq 2^{m-u}+1 {)}_{q}$ & & \cite{LaGuardia:2013B}\\
\hline & $q$ power of an odd prime, $m\geq 3$, &\\
\hline & $r, u \in { \mathbb Z}$, $0< r < u < m$, &\\
\hline & $\displaystyle\sum_{i=u+1}^{m}\left(
\begin{array}{c}
m\\
i\\
\end{array}
\right) > \displaystyle\sum_{i=r+1}^{u}\left(
\begin{array}{c}
m\\
i\\
\end{array}
\right)$, $s_{m}(v)= \displaystyle\sum_{i=0}^{v}\left(
\begin{array}{c}
m\\
i\\
\end{array}
\right)$ & \\
\hline $(2^{m}, 2^{m}- s_{m}(u), \delta ; 2, d_{f}\geq 2^{r+1}
{)}_{q}$ & &\cite{LaGuardia:2013B}\\
\hline & $q$ power of an odd prime, $m\geq 4$, $\delta =
\displaystyle\sum_{i=r+1}^{v}\left(
\begin{array}{c}
m\\
i\\
\end{array}
\right)$, &\\
\hline & $r, u, v \in {\mathbb Z}$, $0 < r < v < u < m$, &\\
\hline & $\displaystyle\sum_{i=u+1}^{m}\left(
\begin{array}{c}
m\\
i\\
\end{array}
\right) \geq \displaystyle\sum_{i=r+1}^{v}\left(
\begin{array}{c}
m\\
i\\
\end{array}
\right) \geq \displaystyle\sum_{i=v+1}^{u}\left(
\begin{array}{c}
m\\
i\\
\end{array}
\right)$, &\\
\hline $(l^{m}, l^{m}-S_{m}(u), S_{m}(u)- S_{m}(r); 1, d_{f}\geq
(b+2)l^{a} {)}_{q}$ & &\cite{LaGuardia:2013B}\\
\hline  & $m\geq 3$, $l\geq 3$ are integers, $q$ prime power, $l | (q - 1)$,  &\\
\hline & $r, u \in { \mathbb Z}$,
$0< r < u < m(l-1)$ &\\
\hline  & $\displaystyle\sum_{i=u+1}^{m}\left(
\begin{array}{c}
m\\
i\\
\end{array}
\right)_{l} \geq\displaystyle\sum_{i=r+1}^{u}\left(
\begin{array}{c}
m\\
i\\
\end{array}
\right)_{l},$ &\\
\hline & $a, b \in { \mathbb Z}$, $r = a(l- 1) +
b$, $0 \leq b \leq l-2$, &\\
\hline  & $S_{m}(v)= \displaystyle\sum_{i=0}^{v}
\displaystyle\sum_{k=0}^{m}{(-1)}^{k}\left(
\begin{array}{c}
m\\
k\\
\end{array}
\right)\left(
\begin{array}{c}
m-1+i-kl\\
m-1\\
\end{array}
\right)$ &\\
\hline
\hline Expanded Character & &\\
\hline $(t2^{m}, t(2^{m}- s_{m}(u)), t(s_{m}(u)- s_{m}(r)) ; 1,
d_{f}\geq
2^{r+1} {)}_{p}$ & & New\\
\hline $(t2^{m}, ts_{m}(u), t(s_{m}(u)- s_{m}(r)) ; \mu,
d_{f}\geq 2^{m-u}+1 {)}_{p}$ & & New\\
\hline & $q=p^{t}$ power of an odd prime, $m\geq 3$, &\\
\hline &
$r, u \in { \mathbb Z}$, $0< r < u < m$ &\\
\hline & $\displaystyle\sum_{i=u+1}^{m}\left(
\begin{array}{c}
m\\
i\\
\end{array}
\right) > \displaystyle\sum_{i=r+1}^{u}\left(
\begin{array}{c}
m\\
i\\
\end{array}
\right)$, $s_{m}(v)= \displaystyle\sum_{i=0}^{v}\left(
\begin{array}{c}
m\\
i\\
\end{array}
\right)$ & \\
\hline $(t2^{m}, t(2^{m}- s_{m}(u)), t\delta ; 2, d_{f}\geq
2^{r+1}{)}_{p}$ & & New\\
\hline & $q=p^{t}$ power of an odd prime, &\\
\hline & $m\geq 4$, $\delta = \displaystyle\sum_{i=r+1}^{v}\left(
\begin{array}{c}
m\\
i\\
\end{array}
\right)$ &\\
\hline & $r, u, v \in { \mathbb Z}$, $0 < r < v < u < m$, &\\
\hline & $\displaystyle\sum_{i=u+1}^{m}\left(
\begin{array}{c}
m\\
i\\
\end{array}
\right) \geq \displaystyle\sum_{i=r+1}^{v}\left(
\begin{array}{c}
m\\
i\\
\end{array}
\right) \geq \displaystyle\sum_{i=v+1}^{u}\left(
\begin{array}{c}
m\\
i\\
\end{array}
\right)$, &\\
\hline $(tl^{m}, t(l^{m}-S_{m}(u)), t(S_{m}(u)- S_{m}(r)); 1,
d_{f}\geq (b+2)l^{a} {)}_{q}$ & & New\\
\hline  & $m\geq 3$, $l\geq 3$ are integers,  &\\
\hline & $q=p^{t}$ prime power, $l | (q - 1)$ &\\
\hline & $r, u \in { \mathbb Z}$,
$0< r < u < m(l-1)$ &\\
\hline  & $\displaystyle\sum_{i=u+1}^{m}\left(
\begin{array}{c}
m\\
i\\
\end{array}
\right)_{l} \geq\displaystyle\sum_{i=r+1}^{u}\left(
\begin{array}{c}
m\\
i\\
\end{array}
\right)_{l},$ &\\
\hline & $a, b \in { \mathbb Z}$, $r = a(l- 1) +
b$, $0 \leq b \leq l-2$, &\\
\hline  & $S_{m}(v)= \displaystyle\sum_{i=0}^{v}
\displaystyle\sum_{k=0}^{m}{(-1)}^{k}\left(
\begin{array}{c}
m\\
k\\
\end{array}
\right)\left(
\begin{array}{c}
m-1+i-kl\\
m-1\\
\end{array}
\right)$ &\\
\hline
\hline Melas & &\\
\hline
\hline $(q^{m}-1, q^{m}-m-1, m; 1, d_{f}\geq 3 {)}_{q}$ &
$q\neq 2$ prime power, $q$ even, $m\geq 2$ & \cite{LaGuardia:2013C}\\
\hline $(2^{m}-1, 2^{m}-m-1, m; 1, d_{f} \geq 5{)}_{2}$ & $m$ odd
& \cite{LaGuardia:2013C}\\
\hline
\hline Expanded Melas & &\\
\hline \hline $(t(q^{m}-1), t(q^{m}-m-1), tm; 1, d_{f}\geq 3
{)}_{p}$ &
$q\neq 2$ prime power, $q=p^{t}$ even, $m\geq 2$ & New\\
\hline
\end{tabular}
\end{center}
\end{table}

\clearpage


\end{document}